\documentclass{fac}

\usepackage{graphicx}
\usepackage{amssymb}
\usepackage{amsfonts}
\usepackage{amsmath}
\usepackage{dsfont}
\usepackage{MnSymbol}
\usepackage{mathtools}
\usepackage{epsfig}
\usepackage{epstopdf}
\usepackage{tikz}
\usepackage{amsthm}
\ifprodtf \else \usepackage{latexsym}\fi

\newcommand\black{\ensuremath{\blacktriangleright}}
\newcommand\white{\ensuremath{\vartriangleright}}

\newif\ifamsfontsloaded
\ifprodtf
  \newcommand\whbl{\white\kern-.1em--\kern-.1em\black}
  \newcommand\blwh{\black\kern-.1em--\kern-.1em\white}
  \newcommand\blbl{\black\kern-.1em--\kern-.1em\black}
  \newcommand\whwh{\white\kern-.1em--\kern-.1em\white}
  \amsfontsloadedtrue
\else
  \checkfont{msam10}
  \iffontfound
    \IfFileExists{amssymb.sty}
      {\usepackage{amssymb}\amsfontsloadedtrue
       \newcommand\whbl{\white\kern-.125em--\kern-.125em\black}%
       \newcommand\blwh{\black\kern-.125em--\kern-.125em\white}%
       \newcommand\blbl{\black\kern-.125em--\kern-.125em\black}%
       \newcommand\whwh{\white\kern-.125em--\kern-.125em\white}}
      {}
  \fi
\fi

\newtheorem{theorem}{Theorem}[section]
\newtheorem{definition}[theorem]{Definition}

\title[Draft of Algebra of Concurrent Games]
      {Algebra of Concurrent Games}

\author[Yong Wang]
    {Yong Wang\\
     College of Computer Science and Technology,\\
     Faculty of Information Technology,\\
     Beijing University of Technology, Beijing, China\\
     }

\correspond{Yong Wang, Pingleyuan 100, Chaoyang District, Beijing, China.
            e-mail: wangy@bjut.edu.cn}

\pubyear{2019}
\pagerange{\pageref{firstpage}--\pageref{lastpage}}

\begin{document}
\label{firstpage}

\makecorrespond

\maketitle

\begin{abstract}
We introduce parallelism into the basic algebra of games to model concurrent game algebraically. Parallelism is treated as a new kind of game operation. The resulted algebra of concurrent games can be used widely to reason the parallel systems.
\end{abstract}

\begin{keywords}
Games; Algebra; Concurrent Games; Formal Theory.
\end{keywords}

\section{Introduction}{\label{int}}

The combination of games and logic (especially computational logic) always exists two ways \cite{LIG}: one is the usual usage of game theory to interpret logic, such as the well-known game semantics; the other is to use logic to understand game theory, such as game logic \cite{GL1} \cite{GL2} \cite{CGL}. The basic algebra of games \cite{BAG1} \cite{BAG2} is also a way to use logic to interpret games algebraically.

Since concurrency is a fundamental activity pattern in nature, it exists not only in computational systems, but also in any process in nature. Concurrent games \cite{CG1} \cite{CG2} \cite{CG3} make the game theory to capture concurrency. But the logic of concurrent games is still missing.

In this paper, we extend the basic algebra of games \cite{BAG1} \cite{BAG2} to discuss the logic of concurrent games algebraically. This paper is organized as follows. In Section \ref{bag}, we repeat the main concepts and conclusions of the basic algebra of games for the convenience of the reader. We give the algebra of concurrent games in Section \ref{acg}. And finally, in Section \ref{con}, we conclude this paper. 

\section{Basic Algebra of Games}\label{bag}

For the convenience of the readers, in this section, we repeat the definitions and main results of Basic Algebra of Games (abbreviated BAG) in \cite{BAG1} and \cite{BAG2}, for they are used in the following sections.

\begin{definition}[Game language]\label{gl}
The game language $GL$ consists of:
\begin{enumerate}
  \item a set of atomic games $\mathcal{G}_{at}=\{g_a\}_{a\in A}$, and a special \emph{idle} atomic game $\iota=g_0\in\mathcal{G}_{at}$;
  \item game operations choice of first player $\vee$, choice of second player $\wedge$, dualization $^d$ and composition of games $\circ$.
\end{enumerate}

Atomic games and their duals are called literals.
\end{definition}

\begin{definition}[Game terms]\label{gt}
The game terms are defined inductively as follows:
\begin{itemize}
  \item every atomic game $g_a$ is a game term;
  \item if $G,H$ are game terms, then $G^d, H^d$, $G\vee H$, $G\wedge H$ and $G\circ H$ are all game terms.
\end{itemize}
\end{definition}

\begin{definition}[Outcome conditions]
$GL=\langle S, \{\rho^i_a\}_{a\in A;i=1,2}\rangle$ are called game boards, where $S$ is the set of states and $\rho^i_a\subseteq S\times P(S)$ are outcome relations, which satisfy the following two forcing conditions:

\begin{enumerate}
  \item monotonicity (MON): for any $s\in S$, and $X\subseteq Y\subseteq S$, if $s\rho^i_a X$, then $s\rho^i_a Y$;
  \item consistency (CON): for any $s\in S, X\subseteq S$, if $s\rho^1_a X$, then not $s\rho^2_a (S-X)$.
\end{enumerate}

And the following optional conditions:

\begin{enumerate}
  \item termination (FIN): for any $s\in S$, then $s\rho^i_a S$, and the class of terminating game boards are denoted $\textbf{FIN}$;
  \item determinacy (DET): $s\rho^2_a (S-X)$ iff $s\rho^1_a X$, and the class of determined game boards are denoted $\textbf{DET}$.
\end{enumerate}

The outcome relation $\rho^i_G$ for any game term $G$ can be defined inductively according to the structure of $G$.
\end{definition}

\begin{definition}[Included]
For game terms $G_1$ and $G_2$ on game board $B$, if $\rho^i_{G_1}\subseteq \rho^i_{G_2}$, then $G_1$ is i-included in $G_2$ on $B$, denoted $G_1\subseteq_i g_2$; if $G_1\subseteq_1 G_2$ and $G_1\subseteq_2 G_2$, then $G_1$ is included in $G_2$ on $B$, denoted $B\models G_1\preceq G_2$; if $B\models G_1\preceq G_2$ for any $B$, then $G_1\preceq G_2$ is called a valid term inclusion, denoted $\models G_1\preceq G_2$.

If $G_1$ and $G_2$ are assigned the same outcome relation in $B$, then they are equivalent on $B$, denoted $B\models G_1=G_2$; if $B\models G_1=G_2$ for any game board $B$, then $G_1=G_2$ is a valid term identity, denoted $\models G_1=G_2$.
\end{definition}

The axioms of BAG are shown in Table \ref{AxiomsOfBAG}.

\begin{center}
    \begin{table}
        \begin{tabular}{@{}ll@{}}
            \hline No. &Axiom\\
            $G1$ & $x\vee x = x\quad x\wedge x =x$\\
            $G2$ & $x\vee y = y\vee x\quad x\wedge y = y\wedge x$\\
            $G3$ & $x\vee (y\vee z)=(x\vee y)\vee z\quad x\wedge (y\wedge z)=(x\wedge y)\wedge z$\\
            $G4$ & $x\vee(x\wedge y)=x\quad x\wedge (x\vee y)=x$\\
            $G5$ & $x\vee(y\wedge z)=(x\vee y)\wedge (x\vee z)\quad x\wedge (y\vee z)=(x\wedge y)\vee (x\wedge z)$\\
            $G6$ & $(x^d)^d = x$\\
            $G7$ & $(x\vee y)^d= x^d\wedge y^d\quad (x\wedge y)^d = x^d\vee y^d$\\
            $G8$ & $(x\circ y)\circ z = x\circ(y\circ z)$\\
            $G9$ & $(x\vee y)\circ z = (x\circ z)\vee (y\circ z)\quad (x\wedge y)\circ z = (x\circ z)\wedge (y\circ z)$\\
            $G10$ & $x^d\circ y^d=(x\circ y)^d$\\
            $G11$ & $y\preceq z \rightarrow x\circ y\preceq x\circ z$\\
            $G12$ & $x\circ \iota = \iota\circ x = x$\\
            $G13$ & $\iota^d=\iota$\\
        \end{tabular}
        \caption{Axioms of BAG}
        \label{AxiomsOfBAG}
    \end{table}
\end{center}

\begin{definition}[Canonical terms]
The canonical terms can be defined as follows:
\begin{enumerate}
  \item $\iota$ and $g_a$ are canonical terms;
  \item $\bigvee_{i\in I}\bigwedge_{k\in K_i}g_{ik}\circ G_{ik}$ is a canonical term, where $g_{ik}$ is a literal and $G_{ik}$ is a game term.
\end{enumerate}
\end{definition}

\begin{theorem}[Elimination theorem 1]\label{et1}
Every game term $G$ is equivalent to a canonical game term in $\textbf{BAG}^\iota$.
\end{theorem}

\begin{definition}[Isomorphic]
Two canonical terms $G$ and $H$, if one can be obtained from the other by means of successive permutations of conjuncts (and disjuncts) within the same $\bigwedge$'s (and $\bigvee$'s) in subterms.
\end{definition}

\begin{definition}[Embedding]
The embedding of canonical terms is defined inductively as follows:

\begin{itemize}
  \item $\iota\rightarrowtail\iota$;
  \item if $g,h$ are literals and $G,H$ are canonical terms, $g\circ G$ embeds into $h\circ H$ iff $g=h$ and $G\rightarrowtail H$;
  \item $\bigwedge_{k\in K}g_k\circ G_k$ embeds into $\bigwedge_{m\in M} h_m\circ H_m$, if for every $m\in M$ there is some $k\in K$ such that $g_k\circ G_k\rightarrowtail h_m\circ H_m$;
  \item for $G=\bigvee_{i\in I}\bigwedge_{k\in K_i}g_{ik}\circ g_{ik}$ and $H=\bigvee_{j\in J}\bigwedge_{m\in M_j}h_{jm}\circ H_{jm}$, $G\rightarrowtail H$ iff every disjunct of $G$ embeds into some disjunct of $H$.
\end{itemize}
\end{definition}

\begin{definition}[Minimal canonical terms]
The minimal canonical terms are defined inductively as follows:
\begin{itemize}
  \item $\iota$ is a minimal canonical term;
  \item for a canonical term $G=\bigvee_{i\in I}\bigwedge_{k\in K_i} g_{ik}\circ G_{ik}$ with all $G_{ik}$ are minimal canonical terms, then $G$ is minimal canonical term if:
      \begin{enumerate}
        \item $\iota^d$ does not occur in $G$;
        \item None of $g_{ik}$ is $\iota$ unless $G_{ik}$ is $\iota$;
        \item No conjunct occurring in $\bigwedge_{k\in K} g_{ik}\circ G_{ik}$ embeds into another conjunct from the same conjunction;
        \item No disjunct in $G$ embeds into another disjunct of $G$.
      \end{enumerate}
\end{itemize}
\end{definition}

\begin{theorem}[Elimination theorem 2]\label{et2}
Every game term $G$ is equivalent to a minimal canonical game term in $\textbf{BAG}^\iota$.
\end{theorem}

\begin{theorem}[Completeness theorem]\label{ct}
The minimal canonical terms $G$ and $H$ are equivalent iff they are isomorphic.
\end{theorem}

\begin{proof}
See the proof in \cite{BAG1} and \cite{BAG2}, the proof in \cite{BAG1} is based on modal logic and that in \cite{BAG2} is in a pure algebraic way.
\end{proof}

\section{Algebra of Concurrent Games}\label{acg}

Concurrent games \cite{CG1} \cite{CG2} \cite{CG3} mean that players can play games simultaneously. For two atomic games $g_a$ and $g_b$ in parallel denoted $g_a\parallel g_b$, the outcomes may be $g_a\circ g_b$, or $g_b\circ g_a$, or they are played simultaneously. Since $g_a$ and $g_b$ may be played simultaneously, $\parallel$ must be treated as fundamental game operation. In this section, we will add $\parallel$ game operation into the basic algebra of games, and the new formed algebra is called Algebra of Concurrent Games, abbreviated ACG.

\begin{definition}[Game language with parallelism]
The new kind of game operation $\parallel$ is added into the game language in Definition \ref{gl}.
\end{definition}

\begin{definition}[Game terms with parallelism]
If $G$ and $H$ are game terms, then $G\parallel H$ is also a game term in Definition \ref{gt}.
\end{definition}

The definitions of Outcome conditions and Included are the same as those in Section \ref{bag}.

For concurrent games, the following axioms in Table \ref{AxiomsOfParallelism} should be added into Table \ref{AxiomsOfBAG}.

\begin{center}
    \begin{table}
        \begin{tabular}{@{}ll@{}}
            \hline No. &Axiom\\
            $CG1$ & $(x\parallel y)\parallel z = x\parallel(y\parallel z)$\\
            $CG2$ & $g_a\parallel (g_b\circ y) = (g_a\parallel g_b)\circ y$\\
            $CG3$ & $(g_a\circ x)\parallel g_b = (g_a\parallel g_b)\circ x$\\
            $CG4$ & $(g_a\circ x)\parallel (g_b\circ y) = (g_a\parallel g_b)\circ (x\parallel y)$\\
            $CG5$ & $(x\vee y)\parallel z = (x\parallel z)\vee (y\parallel z)$\\
            $CG6$ & $x\parallel (y\vee z) = (x\parallel y)\vee (x\parallel z)$\\
            $CG7$ & $(x\wedge y)\parallel z = (x\parallel z)\wedge (y\parallel z)$\\
            $CG8$ & $x\parallel (y\wedge z) = (x\parallel y)\wedge (x\parallel z)$\\
            $CG9$ & $(x\parallel y)^d= x^d\parallel y^d$\\
            $CG10$ & $\iota\parallel x=x$\\
            $CG11$ & $x\parallel\iota = x$\\
        \end{tabular}
        \caption{Axioms of ACG}
        \label{AxiomsOfParallelism}
    \end{table}
\end{center}

\begin{definition}[Canonical terms with parallelism]
The canonical terms can be defined as follows:
\begin{enumerate}
  \item $\iota$ and $g_a$ are canonical terms;
  \item $\bigvee_{i\in I}\bigwedge_{k\in K_i}||_{l\in L}g_{ikl}\circ ||_{l\in L}G_{ikl}$ is a canonical term, where $g_{ikl}$ is a literal and $G_{ikl}$ is a game term.
\end{enumerate}
\end{definition}

\begin{theorem}[Elimination theorem 1 with parallelism]
Every game term $G$ is equivalent to a canonical game term in $\textbf{ACG}^\iota$.
\end{theorem}

\begin{proof}
Similarly to the proof of Theorem \ref{et1} in \cite{BAG1}, it is sufficient to induct on the structure of game terms, and the new case is $\parallel$.
\end{proof}

\begin{definition}[Isomorphic with parallelism]
Two canonical terms $G$ and $H$, if one can be obtained from the other by means of successive permutations of conjuncts (and disjuncts) within the same $\bigwedge$'s (and $\bigvee$'s) in subterms.
\end{definition}

\begin{definition}[Embedding with parallelism]
The embedding of canonical terms is defined inductively as follows:

\begin{itemize}
  \item $\iota\rightarrowtail\iota$;
  \item if $g_1, g_2,h_1,h_2$ are literals and $G,H$ are canonical terms, $(g_1\parallel g_2)\circ G$ embeds into $(h_1\parallel h_2)\circ H$ iff $g_1=h_1$, $g_2=h_2$, and $G\rightarrowtail H$;
  \item if $g,h$ are literals and $G,H$ are canonical terms, $g\circ G$ embeds into $h\circ H$ iff $g=h$ and $G\rightarrowtail H$;
  \item $\bigwedge_{k\in K}g_k\circ G_k$ embeds into $\bigwedge_{m\in M} h_m\circ H_m$, if for every $m\in M$ there is some $k\in K$ such that $g_k\circ G_k\rightarrowtail h_m\circ H_m$;
  \item for $G=\bigvee_{i\in I}\bigwedge_{k\in K_i}g_{ik}\circ g_{ik}$ and $H=\bigvee_{j\in J}\bigwedge_{m\in M_j}h_{jm}\circ H_{jm}$, $G\rightarrowtail H$ iff every disjunct of $G$ embeds into some disjunct of $H$.
\end{itemize}
\end{definition}

\begin{definition}[Minimal canonical terms with parallelism]
The minimal canonical terms are defined inductively as follows:
\begin{itemize}
  \item $\iota$ is a minimal canonical term;
  \item for a canonical term $G=\bigvee_{i\in I}\bigwedge_{k\in K_i} ||_{l\in L_k}g_{ikl}\circ G_{ikl}$ with all $G_{ikl}$ are minimal canonical terms, then $G$ is minimal canonical term if:
      \begin{enumerate}
        \item $\iota^d$ does not occur in $G$;
        \item None of $g_{ikl}$ is $\iota$ unless $G_{ikl}$ is $\iota$;
        \item No conjunct occurring in $\bigwedge_{k\in K} ||_{l\in L_k}g_{ikl}\circ G_{ikl}$ embeds into another conjunct from the same conjunction;
        \item No disjunct in $G$ embeds into another disjunct of $G$.
      \end{enumerate}
\end{itemize}
\end{definition}

\begin{theorem}[Elimination theorem 2]
Every game term $G$ is equivalent to a minimal canonical game term in $\textbf{ACG}^\iota$.
\end{theorem}

\begin{proof}
Similarly to the proof of Theorem \ref{et2} in \cite{BAG1}, it is sufficient to induct on the structure of game terms according to the definition of minimal canonical game terms, the new case is $\parallel$.
\end{proof}

Since the Elimination theorem holds, the completeness of ACG only need to discuss the relationship among the minimal canonical terms.

\begin{theorem}[Completeness theorem]
The minimal canonical terms $G$ and $H$ are equivalent iff they are isomorphic.
\end{theorem}

\begin{proof}
Similarly to the proof of Theorem \ref{ct} in \cite{BAG1}:

\begin{enumerate}
  \item firstly, translate ACG to the same modal logic in \cite{BAG1};
  \item secondly, prove the translation preserve validity;
  \item Finally, get the completeness result.
\end{enumerate}

The only difference is the translation of $\parallel$, because two game terms $G_1$ and $G_2$ may be in race condition, denoted $G_1 \% G_2$, for a game term $G$ and the corresponding modal logic formula $m(G)$, the translation of $\parallel$ is:

\begin{enumerate}
  \item if $G_1\% G_2$, then $m(G_1\parallel G_2)=m(G_1)(m(G_2))$ or $m(G_1\parallel G_2)=m(G_2)(m(G_1))$;
  \item else $m(G_1\parallel G_2)=m(G_1)\vee m(G_2)$.
\end{enumerate}
\end{proof}

\section{Conclusions}\label{con}

We introduce parallelism in the basic algebra of games \cite{BAG1} \cite{BAG2} to model concurrent game algebraically. The resulted algebra ACG can be used reasoning parallel systems with game theory support.

\newpage

%

\label{lastpage}

\end{document}